\documentclass[12pt,paper]{article}

\usepackage{amsmath,amssymb,amsthm,amsfonts,fullpage}
\usepackage{enumitem}  % for using the left marging in itemize
\usepackage{biblatex}
\usepackage{authblk}
\usepackage[hidelinks,pdfauthor={Abdelghaffar Chibloun}]{hyperref}
\hypersetup{
    colorlinks = true,
    citecolor  = blue,
    linkcolor  = blue,
    urlcolor   = blue
}
\usepackage{relsize}
\usepackage{color,xcolor,soul}
\providecommand{\keywords}[1]
{{\small\textbf{{Keywords--}} #1}}
\providecommand{\msc}[1]
{{\small\textbf{{Mathematics Subject Classification (2020) --}} #1}}

\theoremstyle{plain}
\newtheorem{theorem}{Theorem}[section]
\newtheorem{definition}{Definition}
\newtheorem{proposition}{Proposition}
\newtheorem{corollary}{Corollary}
\newtheorem{lemma}{Lemma}
\newtheorem{remark}{Remark}
\newtheorem{example}{Example}

\DeclareMathOperator{\aut}{Aut}

\DeclareMathOperator{\id}{Id}

\title{{Skew polycyclic over finite chain rings associated to trinomials}}

\author[1]{Maryam Bajalan\thanks{E-mail: \rm{maryam.bajalan@math.bas.bg}. Supported by the National Science Fund, Ministry of Education and Science of Bulgaria, under contract KP-06-N92/2.}}
\author[2]{Edgar Mart\'inez-Moro\thanks{E-mail: \rm{edgar.martinez@uva.es}. Partially supported by Grant  SAGACT-1 MCIN/AEI/ 10.13039/501100011033 y FEDER "Una manera de hacer Europa" PID2022-138906NB-C21 2023/2027}}
\author[3]{Hassan  Ou-azzou\thanks{E-mail: \rm{hassan.ouazzou-ext@um6p.ma}}}

 \affil[1]{Institute of Mathematics and Informatics, Bulgarian Academy of Sciences, Bl. 8,
Acad. G. Bonchev Str., 1113, Sofia, Bulgaria}
  \affil[2]{ Institute of Mathematics, University of Valladolid, Castilla, Spain}
 \affil[3]{College of Computing, Mohammed VI Polytechnic University, Ben Guerir, Morocco}

\date{}
\begin{document}

	\maketitle
\begin{abstract}
This work studies skew polycyclic codes over finite chain rings defined by central trinomials. For this class of codes, we investigate Hamming equivalence in the non-commutative (skew) setting. We introduce an equivalence relation on the defining trinomials and demonstrate that it admits a group-theoretic characterization in terms of a group of binomials equipped with the Schur multiplication. We determine the conditions under which skew polycyclic codes are Hamming equivalent to those defined by the specific trinomial $x^n-(x^\ell+1)$. This reduces the classification problem for these codes, up to Hamming equivalence, to a canonical case. Finally, we determine the size of the corresponding equivalence class using the decomposition of the unit group of the underlying chain ring.

\end{abstract}

    
\msc{94B05, 94B15,
16S36}

% ====================================================================
% Introduction
% ====================================================================  
\section{Introduction}	
A central problem in coding theory is constructing linear codes with optimal parameters. Given the sheer amount of linear codes, an exhaustive search for optimal ones is almost always practically infeasible. We can work around this problem by limiting our search to a subclass of linear codes with additional algebraic properties. One such subclass is polycyclic codes. While polycyclic codes have been defined since the very beginning of coding theory in the 70s, they have recently gained significant attention over both finite fields and finite rings (see, for example, \cite{EdgarMaryamSteve, EdgarMaryamSteve2, Sole2, fotue2020polycyclic, sole, OUAZZOU2025114279}). In particular, the study of codes associated with trinomials over finite fields was considered in \cite{aydin2022polycyclic} and further developed in \cite{Equiv2025, chibloun2025polycyclic}, showcasing many interesting properties and applications that warrant extension to the more general case of skew polynomials over chain rings, in other words, the class of \textit{skew polycyclic codes}. In this framework, a code can be seen as a left ideal of the skew polynomial ring modulo a left principal ideal generated by an associated polynomial.

One way to further reduce the task of finding new codes is to introduce an equivalence relation that preserves important properties and characteristics of a code, such as its weight distribution. In the context of linear codes with respect to the Hamming distance, monomial equivalence has been widely adopted as a suitable equivalence relation. This goes back to the early work of MacWilliams, where it was shown that over finite fields, such a relation essentially amounts to having a linear weight-preserving isomorphism. Unfortunately, checking whether two linear codes are equivalent is generally computationally challenging \cite{sendrier2013easy}. However, harnessing the additional algebraic structure present in skew polycyclic codes can help to overcome this computational limitation. An example of this was proposed in \cite{CHEN20121217}, where the authors introduced an equivalence relation called $n$-isometry to classify constacyclic codes of length $n$ over $\mathbb{F}_q$.  The fact of being in the same class with respect to the  $n$-isometry  implies the existence of an isometry with respect to the Hamming distance between the base rings where the codes are defined. This induces a biunivocal relation between the set of constacyclic codes on those ambient spaces and the ambient spaces themselves, preserving the weight distribution and the algebraic structure. In \cite{CHEN201460}, the authors provided an explicit mapping between generator polynomials. These relations have been generalized to constacyclic codes over finite chain rings \cite{chibloun2024isometry} and skew constacyclic codes over finite fields \cite{OUAZZOU2025114279}. These types of ideas have been generalized to polycyclic codes given by a trinomial in \cite{shi2023equivalence,chibloun2025polycyclic}. 

In addition, in the recent work \cite{pumpluen2025using}, the author studies monomial homomorphisms 
and isomorphisms $G_{\tau,\alpha,k}$,  defined by
$$G_{\tau,\alpha,k}\left(\sum_{i=0}^{n-1} c_i t^i\right) = \sum_{i=0}^{n-1} \tau(c_i)(\alpha t^k)^i,$$
where $\tau \in \mathrm{Aut}(S)$, $\alpha \in S^\times$, and $k \in \mathbb{N}$, and the map is defined over quotients 
of skew polynomial rings $S[t;\sigma]/S[t;\sigma]f(x)$ onto $S[t;\sigma]/S[t;\sigma]h(x)$, with $S$ a unital associative ring 
not necessarily commutative, and the polynomials $f, h$ are not necessarily central. The code families in  each ambient space $C_f$ and $C_h$ are 
called \textit{isometric} if such a map exists,  and they are called \textit{Chen isometric} if $\tau = \mathrm{id}$, 
\textit{equivalent} if $k=1$, and \textit{Chen equivalent} if $\tau = \mathrm{id}$ and $k=1$. 
When $S$ is a finite field, the size of each equivalence class is explicitly determined along 
with combinatorial examples. In Section~8 of \cite{pumpluen2025using}, for constacyclic codes 
over finite chain rings, it is shown that $|N_n^\sigma(1+xS)|$ determines the exact number of 
Chen equivalence classes, though the authors note that this group seems difficult to compute in 
general and only estimates on its size are provided.

In the present paper, we complement this work in a different direction. The map
$$\varphi_\alpha\left(\sum_{i=0}^{n-1} g_i x^i\right) = \sum_{i=0}^{n-1} g_i N_i^\sigma(\alpha) x^i$$
in Definition~\ref{Def_Iso} coincides exactly with the Chen equivalence map $G_{\mathrm{id},\alpha,1}$, so our notation 
$a(x) \sim_{(n,\sigma)} b(x)$ corresponds precisely to Chen equivalence. Our notion of Hamming 
$(n,\sigma)$-equivalence is, however, broader: two codes $C_1 \subseteq R^r_{f_1}$ and 
$C_2 \subseteq R^r_{f_2}$ are Hamming $(n,\sigma)$-equivalent if there exists any 
Hamming-distance-preserving ring isomorphism $\varphi: R^r_{f_1} \to R^r_{f_2}$ with 
$\varphi(C_1) = C_2$, without requiring that $\varphi$ take the special form $\varphi_\alpha$. 
While Chen equivalence implies Hamming $(n,\sigma)$-equivalence, the converse is not trivial. 
In Theorem~4.1, we prove the converse for the special case  of polynomials of the form $f_2(x) = x^n - (x^\ell + 1)$, 
and exploit this to connect the equivalence classes to the groups $H_{\ell,\sigma}$ and 
$B_{\ell,\sigma}$. Finally, while \cite{pumpluen2025using} leaves open the explicit computation 
of the number of Chen equivalence classes of constacyclic codes over chain rings, in Section~5 
we explicitly compute $|\ker(\theta)| = |\ker(\theta|_{T^*})| \cdot |\ker(\theta|_U)|$,  where $U=1+\gamma R$ and $T^*$ is Teichm\"uller set of the chain ring  $R$ excluding the zero element. Thus,  we can 
determine the exact number of Chen equivalence classes of trinomial polycyclic codes that are 
not constacyclic over finite chain rings.

Thus, in this paper, we focus our attention on skew polycyclic codes over finite chain rings associated with
trinomials, which is not just an incremental step in the research of those equivalences, but requires some new theoretical background as well as a refinement of the notion of equivalence with respect to the non-skew case.  

The outline of the paper is as follows. In Section~\ref{secPrelim}, we provide the basic facts on finite chain rings and skew polynomials over finite chain rings, and we fix the notation. Section~\ref{sec:SPCFCC} provides a description of skew polycyclic codes over finite chain rings, whereas Section~\ref{eq} deals with the Hamming equivalence of those objects. Finally, in Section~\ref{sec:number}, we give a way to count the number of Hamming equivalence classes.

% ====================================================================
% Preliminaries
% ====================================================================   
\section{Preliminaries}\label{secPrelim}
In this section, we recall basic properties of finite chain rings and introduce skew polynomial rings over such rings. 
% ====================================================================
% Finite chain rings
% ====================================================================
\subsection{Finite chain rings}\label{Finite chain rings}

A \emph{chain ring} is a ring whose ideals are totally ordered by inclusion. 
Throughout this paper, $R$ will denote a finite commutative chain ring with identity. % It is well known that $R$ is a finite local principal ideal ring. 
For general background on finite chain rings, we refer the reader to \cite{clark1973enumeration, mcdonald1974finite}. Let $\mathfrak{m}$ denote the unique maximal ideal of $R$, and let $\mathbb{F}_q = R/\mathfrak{m}$ be its residue field, where $q=p^r$ is a power of a prime $p$. We will fix a generator of $\mathfrak{m}$ and denote it by $\gamma$, and $e$ will denote its nilpotency index, so that $\gamma^e=0$ and $\gamma^{e-1} \neq 0$. Then, the ideals of $R$ form the chain
\begin{align}
    \langle 0 \rangle = \langle \gamma^e \rangle \subsetneq \langle \gamma^{e-1} \rangle \subsetneq \cdots \subsetneq \langle \gamma \rangle = \mathfrak{m} \subsetneq \langle 1 \rangle = R.
\end{align} The map
$\overline{\, \cdot \,}\colon R \to \mathbb{F}_q$ will denote the canonical projection of $R$ onto its residue field. It can be extended coefficient-wise to a ring homomorphism $R[x] \to \mathbb{F}_q[x]$, and for a polynomial $h \in R[x]$ we denote its image by $\overline{h}$. 
A polynomial $h \in R[x]$ is called \emph{basic irreducible} if $\overline{h}$ is irreducible in $\mathbb{F}_q[x]$.

An important class of chain rings is the class of \emph{Galois rings}, defined as $\mathrm{GR}(p^{m},r):=\mathbb{Z}_{p^m}[x]/\langle h\rangle$, where $h\in \mathbb{Z}_{p^m}[x]$ is a basic irreducible polynomial of degree $r$. Then $\mathrm{GR}(p^m,r)$ contains a unit $\omega$ of multiplicative order $p^r-1$ such that $\mathrm{GR}(p^m,r)=\mathbb{Z}_{p^m}[\omega]$.
 
The unit group $R^\times$ contains a unique cyclic subgroup of order $q-1$, given by $T^*:=\{\xi\in R^\times \colon \xi^{q-1}=1\}$. The set $T:=T^*\cup\{0\}$ is called the \emph{Teichm\"uller set} of $R$. In the special case $R=\mathrm{GR}(p^m,r)$, this set is given explicitly by $T=\{0,1,\omega,\omega^2,\dots,\omega^{p^r-2}\}$. The set $T^*$ forms a complete set of representatives of the residue field $\mathbb{F}_q$ under the canonical projection $\overline{\, \cdot \,}\colon R \to \mathbb{F}_q$. Every element $a\in R$ admits a unique expansion $a=\xi_0+\gamma \xi_1+\cdots+\gamma^{e-1}\xi_{e-1}$, where $\xi_i\in T$. This yields the unique decomposition of the   group of units of $\mathrm{GR}(p^{m},r)$ as
\begin{align}\label{UnitGroup}
R^\times = T^*\cdot (1+\gamma R)\cong T^*\times (1+\gamma R).
\end{align}
Since $|R^\times|=p^{r(e-1)}(p^r-1)$ and $|T^*|=p^r-1$, it follows that the subgroup $U:=1+\gamma R$ has order $p^{r(e-1)}$. Hence $U$ is a finite abelian $p$-group since $R$ is commutative. Moreover, if $x\in T^*\cap U$, then the order of $x$ divides both $|T^*|$ and $|U|$. Since $\gcd(p^r-1,p^{r(e-1)})=1$, we must have $x=1$. Therefore, $T^*\cap U=\{1\}$.
The following result will be useful in the following sections.

\begin{lemma}\label{stabilize}
Let $\sigma \in \operatorname{Aut}(R)$. Then  $\sigma(T^*) = T^*$ and  $\sigma(U) = U$.
\end{lemma}

\begin{proof} Note that
$\sigma$ maps maximal ideals to maximal ideals. %In particular, $\sigma(\mathfrak{m})$ is a maximal ideal of $R$. 
Since $\mathfrak{m}$ is the unique maximal ideal of $R$, it follows that $\sigma(\mathfrak{m}) = \mathfrak{m}$. As $\mathfrak{m} = \langle \gamma \rangle$, the equality $\sigma(\mathfrak{m}) = \mathfrak{m}$ implies that $\langle \sigma(\gamma) \rangle = \langle \gamma \rangle$. In a chain ring, an element $x \in \mathfrak{m}$ generates $\mathfrak{m}$ if and only if $x \notin \mathfrak{m}^2$. Hence, $\sigma(\gamma) \notin \mathfrak{m}^2$, and therefore $\sigma(\gamma) = \gamma u$ for some unit $u \in R^\times$. It follows that $\sigma(\gamma R) = \gamma R$, and consequently $\sigma(1 + \gamma R) = 1 + \gamma R$. Moreover, for any $\xi \in T^*$, the relation $\xi^{q-1} = 1$ implies that $\sigma(\xi)^{q-1} = 1$, and hence $\sigma(\xi) \in T^*$. Therefore, $ \sigma(T^*) \subseteq T^*$ and $\sigma(U) \subseteq U$. Now, since $\sigma$ is a bijection, those inclusions are equal.
\end{proof}
% ====================================================================
% Skew polynomials over finite chain rings
% ====================================================================
\subsection{Skew polynomials over finite chain rings}
In this section, we briefly recall some notation and results for skew polynomial rings; see \cite{Bueso2013, Lam1992homomorphisms, Lam1988algebraic, mcdonald1974finite} for further details. Let $R$ be a finite commutative chain ring with maximal ideal $\mathfrak{m}$ and residue field $\mathbb{F}_q = R/\mathfrak{m}$, and let $\sigma \in \operatorname{Aut}(R)$. The $\sigma$-skew polynomial ring $R[x;\sigma]$ is given by the set of polynomials $h(x)=\sum_{i=0}^n a_i x^i$, where $a_i \in R, \ n \ge 0$, equipped with the usual addition and multiplication determined by
\begin{align}
    xa=\sigma(a)x,
\end{align}
and extended to all polynomials by distributivity. For convenience, we often write $h$ instead of $h(x)$ to denote a polynomial in $R[x;\sigma]$. Note that the ring $R[x;\sigma]$ is non-commutative unless $\sigma = \id_R$.
Since $\sigma(\mathfrak m)=\mathfrak m$, $\sigma$ induces a well-defined automorphism $\overline{\sigma}$ of $\mathbb{F}_q$, defined by
\begin{align}
    \overline{\sigma}(\overline{a}) = \overline{\sigma(a)}, \quad \text{for all}\, a \in R.
\end{align}
Thus, the canonical projection $\overline{\, \cdot \,}\colon R\rightarrow  \mathbb F_q$ extends to a surjective ring homomorphism
\begin{align}
    \overline{\, \cdot \,}\colon  R[x;\sigma] &\longrightarrow  \mathbb F_q[x;\overline{\sigma}] \nonumber\\
  h(x)=\sum_{i=0}^n a_i x^i &\longmapsto \overline{h}(x)=\sum_{i=0}^n \overline{a_i} x^i.
\end{align}
%A polynomial  $h \in R[x;\sigma]$ is said to be \emph{basic irreducible} if its image $\overline{h}$ is irreducible in $\mathbb F_q[x ; \overline{\sigma}]$, and regular if $\overline{h} \neq 0$. 
A polynomial $h \in R[x;\sigma]$ is said to be \emph{central} if it belongs to the center $Z(R[x;\sigma])$ of the ring. The \emph{fixed subring} of $R$ under $\sigma$ is $R^{\sigma}:=\{\,a\in R \mid \sigma(a)=a\,\}$. Let $\sigma \in \operatorname{Aut}(R)$ be of finite order $\mu \neq 1$, and let $n=\mu m$. We define the following set
\begin{align}
\mathcal{A}_{\sigma} := \Big\{ (d_0, \dots, d_{n-1}) \in \big({R^\sigma}\big)^n \,:\, d_i = 0\,\, \text{for}\,\, i \not\equiv 0 \pmod{\mu} \Big\}.
\end{align}
Thus, the elements of $\mathcal{A}_{\sigma}$ are vectors over $R^{\sigma}$ whose support is contained in the set $\{i\colon i\equiv 0 \pmod{\mu}\}$.

\begin{lemma}\label{m20}
A monic polynomial $h(x) = x^n - \sum_{i=0}^{n-1} a_i x^i \in R[x; \sigma]$ is  central (i.e. an element in the center) if and only if
\begin{enumerate}
    \item $\sigma(a_i) = a_i$ for all $0 \le i \le n-1$;
    \item $a_i\big(r - \sigma^i(r)\big) = 0$ for all $r \in R$ and $0 \le i \le n-1$;
    \item $\sigma^n = \id_R$.
\end{enumerate}
\end{lemma}

\begin{proof}
An element in $R[x; \sigma]$ is central if and only if it commutes with the generators of the ring, namely $x$ and each $r \in R$. First, the condition $xh = hx$ provides
\begin{align}
    x^{n+1} - \sum_{i=0}^{n-1} \sigma(a_i)x^{i+1}=x^{n+1} - \sum_{i=0}^{n-1} a_i x^{i+1},
\end{align}
and hence $\sigma(a_i) = a_i$ for all $0 \le i \le n-1$. Next, for $r \in R$, the equality $rh = hr$ yields
\begin{align}
    r x^n - \sum_{i=0}^{n-1} r a_i x^i=\sigma^n(r) x^n - \sum_{i=0}^{n-1} a_i\sigma^i(r) x^i.
\end{align}
Comparing coefficients, we obtain $r = \sigma^n(r)$ for all $r \in R$, and thus $\sigma^n = \id_R$. Moreover, for each $i$, $r a_i = a_i \sigma^i(r)$, which is equivalent to $a_i\big(r - \sigma^i(r)\big) = 0$.
\end{proof}

\begin{lemma}\label{Cor_Asigma}
Let $\sigma \in \operatorname{Aut}(R)$ be an automorphism of finite order $\mu$, and let $n=\mu m$. Let $h(x)=x^n - \sum_{i=0}^{n-1} a_i x^i \in R[x;\sigma]$ be a monic polynomial.
\begin{enumerate}
    \item If $(a_0, \dots, a_{n-1}) \in \mathcal{A}_\sigma$, then $h$ is central.
    \item Suppose additionally that the induced automorphism $\overline{\sigma}$ on the residue field $\mathbb{F}_q$ has order $\mu$. Then $h$  is central if and only if $(a_0, \dots, a_{n-1}) \in \mathcal{A}_\sigma$.  
\end{enumerate}
\end{lemma}

\begin{proof}
Let $(a_0, \dots, a_{n-1}) \in \mathcal{A}_\sigma$. Then each $a_i\in R^\sigma$, and $a_i=0$ for $i\not\equiv 0\pmod{\mu}$. Moreover, if $i\equiv 0\pmod{\mu}$, then
$\sigma^i=\operatorname{Id}_R$. Hence, for all $r \in R$ and all $i$, we have $a_i(r - \sigma^i(r)) = 0$ for all $r \in R$. Also, since $n=\mu m$, we have $\sigma^n=\operatorname{Id}_R$. Thus, $h$ satisfies all conditions of Lemma~\ref{m20}, and hence it is central. To prove the second item, assume that $h$ is central. By Lemma~\ref{m20}, we have $\sigma(a_i) = a_i$ for all $i$, 
\begin{align}\label{del}
    a_i\bigl(r - \sigma^i(r)\bigr) = 0 \quad \text{for all } r \in R.
\end{align}
Let $i\not\equiv 0 \pmod{\mu}$. Since $\overline{\sigma}$ is of order $\mu$,
we have $\overline{\sigma}^i\neq\operatorname{Id}$. Hence, there exists $\overline{r}\in \mathbb F_q$ such that $\overline{r}-\overline{\sigma}^i(\bar r)\neq 0$. Choosing a lift $r\in R$, we obtain that $r-\sigma^i(r)$ is not in $\mathfrak m$, hence is a unit of $R$. Thus, if $i\not\equiv 0\pmod{\mu}$, the equation \eqref{del} implies that $a_i=0$. Since each $a_i$ is fixed by $\sigma$, we conclude that $(a_0,\dots,a_{n-1})\in\mathcal A_\sigma$.
\end{proof}
The following example shows that if $\overline{\sigma}$ does not have order $\mu$, the second item in the above lemma may fail.
\begin{example}\label{ex:converse-fail}
Consider the chain ring $R = \mathbb{F}_3[u] / \langle u^3 \rangle$, and define $\sigma \in \operatorname{Aut}(R)$ by $\sigma(u) = -u$, fixing $\mathbb{F}_3$. Then $\sigma$ has order $\mu = 2$. The monic polynomial $h(x)=x^{2}-u^{2}x \in R[x;\sigma]$ is central because the conditions of Lemma \ref{m20} hold. On the other hand, since $\mu = 2$, the set $\mathcal{A}_\sigma$ consists of all the pairs $(d_0, d_1) \in (R^\sigma)^2$ with $d_1 = 0$. But  the coefficient vector of $h$ is $(0, -u^2)$, which is not in $\mathcal{A}_\sigma$.
\end{example}
The ring $R[x;\sigma]$ admits right (resp. left) division algorithms. More precisely, let $h,k\in R[x;\sigma]$ with $k$ having as leading coefficient a unit in the ring $R$\footnote{Over a general ring, the division algorithm requires the leading coefficient of the divisor to be a unit, since cancellation of leading terms need not be possible otherwise.}. Then, there exist unique polynomials $q,r\in R[x;\sigma]$ such that $h = qk + r$ with $\deg(r)<\deg(k)$. The polynomial $r$ is called the right (resp. left) remainder. For $\beta\in R$, the $i$th $\sigma$-norm of $\beta$ is defined recursively by
\begin{align}
    N_0^\sigma(\beta)=1, \quad N_i^\sigma(\beta)=\sigma\bigl(N_{i-1}^\sigma(\beta)\bigr)\beta \quad (i\ge 1).
\end{align}
As in \cite[Lemma 2.4]{lam1988vandermonde},  $x-\beta$ is a right divisor of $h(x)=\sum_{i=0}^n a_i x^i$ in $R[x;\sigma]$ if and only if $\sum_{i=0}^n a_i N_i^\sigma(\beta)=0$. In this case, $\beta$ is called a right root of $h$.
The following results present some properties of the $\sigma$-norm that will be required in the sequel. 
\begin{lemma}\label{norm-pro}
Let $\alpha,\,\beta\in R$ and let $i,j\geqslant 0$. Then, the following holds:
\begin{enumerate}
    \item $N_i^\sigma(\beta) =\prod_{t=0}^{i-1}\sigma^t(\beta)$; 
    \item If $\beta\in R^\times$, then $N_i^\sigma(\beta)\in R^\times$, and $N_i^\sigma(\beta^{-1})=\bigl(N_i^\sigma(\beta)\bigr)^{-1}$;
    \item $N_i^\sigma(\alpha\beta)=N_i^\sigma(\alpha)N_i^\sigma(\beta)$;
    \item $N_{i+j}^\sigma(\beta)=\sigma^j\bigl(N_i^\sigma(\beta)\bigr)\,N_j^\sigma(\beta)=N_i^\sigma\bigl(\sigma^j(\beta)\bigr)\,N_j^\sigma(\beta)$; 
    \item $\sigma^i(\beta)N_i^\sigma(\beta)=\sigma\bigl(N_i^\sigma(\beta)\bigr)\beta=N_{i+1}^\sigma(\beta)$;
    \item $(\beta x)^i=N_i^\sigma(\beta)x^i$. %Proof by induction and the previous item
    \item $N_m^\sigma(\beta)\in R^\sigma$ if $\sigma^m=\id_R$. % because \begin{align}\sigma\bigl(N_m^\sigma(\beta)\bigr)=\prod_{t=0}^{m-1}\sigma^{t+1}(\beta)=\sigma^m(\beta)\prod_{t=1}^{m-1}\sigma^t(\beta).\end{align}
\end{enumerate}
\end{lemma}

\begin{lemma}\label{hom_le}
Let $\alpha\in R$. Then, the mapping $R[x;\sigma]\to R[x;\sigma]$ defined by $x \mapsto \alpha x$ is a ring homomorphism.
\end{lemma}
\begin{proof} 
For the monomials $g(x)=ax^j$ and $h(x)=bx^k$ in $R[x;\sigma]$, we have
\begin{align*}
gh(\alpha x) & = a \sigma^j(b) N_{j+k}^\sigma(\alpha) x^{j+k}  = a \sigma^j(b) \left( N_j^\sigma(\alpha) \sigma^j(N_k^\sigma(\alpha)) \right) x^{j+k} \\
&= a N_j^\sigma(\alpha) \sigma^j(b) \sigma^j(N_k^\sigma(\alpha)) x^{j+k} = a N_j^\sigma(\alpha) \sigma^j(b N_k^\sigma(\alpha)) x^{j+k} \\
&= (a N_j^\sigma(\alpha) x^j) (b N_k^\sigma(\alpha) x^k)  = g(\alpha x) h(\alpha x),
\end{align*} and the result follows.
\end{proof}
% ====================================================================
% Skew polycyclic codes over finite chain rings
% ====================================================================
\section{Skew polycyclic codes over finite chain rings}\label{sec:SPCFCC}

Let $R$ be a finite commutative chain ring. A \emph{linear code} $C$ of length $n$ over $R$ is an $R$-submodule of $R^n$. Let $f$ denote a monic polynomial of degree $n$ of the form
\begin{align}
    f(x)= x^n-a(x) := x^n - \sum_{i=0}^{n-1} a_i x^i \in R[x; \sigma].
\end{align}
Let $\vec{a} = (a_0, a_1, \dots, a_{n-1}) \in R^n$ be the coefficient vector of $f$. A linear code $C \subseteq R^n$ is called \emph{skew $(\sigma, f)$-polycyclic} (or equivalently, \emph{skew $(\sigma, \vec{a})$-polycyclic}) if it is invariant under the skew polycyclic shift; that is, for every $\vec{c}=(c_0, c_1, \dots, c_{n-1}) \in C$, the vector
\begin{align}
\big(0, \sigma(c_0), \sigma(c_1), \dots, \sigma(c_{n-2})\big) + \sigma(c_{n-1})\vec{a}
\end{align}
also belongs to $C$. If $\sigma = \id$, this reduces to classical $f$-polycyclic codes over rings, see for example \cite{fotue2020polycyclic, EdgarMaryamSteve,EdgarMaryamSteve2}.

An additive map $T: R^n \to R^n$ is called $\sigma$-\emph{semilinear} if $T(\alpha v) = \sigma(\alpha) T(v)$ for all $\alpha \in R$ and $v \in R^n$. By \cite[Lemma 2]{bajalan2023sigma}, $C \subseteq R^n$ is a skew $(\sigma,f)$-polycyclic code if and only if it is invariant under the transformation
\begin{align}
T_{\sigma,\vec{a}}(c_0,\dots,c_{n-1}) = \big(0, \sigma(c_0), \dots, \sigma(c_{n-2})\big) + \sigma(c_{n-1})\vec{a}.
\end{align}
Writing $\sigma(\vec{c}) = (\sigma(c_0),\dots,\sigma(c_{n-1}))$, we have $T_{\sigma,\vec{a}}(\vec{c}) = \sigma(\vec{c}) C_f$,
where $C_f$ is the companion matrix of the polynomial $f(x)$, namely
\begin{align}
C_f =
\begin{pmatrix}
0 & 1 & 0 & \cdots & 0 \\
0 & 0 & 1 & \cdots & 0 \\
\vdots & \vdots & \ddots & \ddots & \vdots \\
0 & 0 & 0 & \ddots & 1 \\
a_0 & a_1 & a_2 & \cdots & a_{n-1}
\end{pmatrix}.
\end{align}
Thus, the code $C\subset R^n$ is skew $(\sigma,f)$-polycyclic if and only if $T_{\sigma,\vec{a}}(\vec{c}) \in C$ for all $\vec{c} \in C$.

The ring $R[x;\sigma]$ admits both right and left division algorithms. Let $\mathcal{R}_f^r$ (resp.\ $\mathcal{R}_f^l$) denote the set of right (resp.\ left) remainders upon division by $f(x)$. Then every $g\in R[x;\sigma]$ can be written uniquely as $g = qf + r$, where $\deg r < \deg f$, and hence $\mathcal{R}_f^r$ forms a complete system of representatives for the quotient ring
\begin{align}
  \frac{R[x;\sigma]}{\langle f\rangle_l},\quad\text{where}\, \langle f\rangle_l = R[x;\sigma]f.
\end{align}
Thus $R[x;\sigma]/\langle f\rangle_l \cong \mathcal{R}_f^r$ as left $R[x;\sigma]$-modules. When needed, we will write $\mathcal{R}_{x^n-a(x)}^r$ to emphasize the defining polynomial $f(x)=x^n-a(x)$. We equip $\mathcal{R}_f^r$ with the usual multiplication 
\begin{align}
    g\cdot h:=\operatorname{rem}_r(gh,f),\quad g,h\in \mathcal{R}_f^r,
\end{align}
where $\operatorname{rem}_r(gh,f)$ denotes the right remainder of $gh$ upon division by $f$. For $g,h,k \in \mathcal{R}_f^r$ with $g\cdot h = q_1 \cdot f + r_1$ and $h\cdot k = q_2\cdot f + r_2$, we have that
\begin{align}\label{associator}
       (g\cdot h)\cdot k - g\cdot (h\cdot k) = \operatorname{rem}_r(g \cdot q_2\cdot f - q_1\cdot f\cdot k, f).
\end{align}
Note that, in general, $f\cdot k \ne k\cdot f$, thus $\mathcal{R}_f^r$ is, in general, non-associative. If the polynomial $f$ is central, then $f\cdot k = k\cdot f$ for all $k \in R[x;\sigma]$, and the associator in \eqref{associator} vanishes. Hence $\mathcal{R}_f^r$ is associative and coincides with the quotient ring $R[x;\sigma]/\langle f\rangle_l$. Define the left $R$-module isomorphism $\Omega_f^r: R^n \to \mathcal{R}_f^r$ by
\begin{align}
    \Omega_f^r(g_0,\dots,g_{n-1}) = \sum_{i=0}^{n-1} g_i x^i.
\end{align}
We equip $\mathcal{R}_f^r$ with the \emph{Hamming weight}
\begin{align}
    \operatorname{wt}_H(g(x)) = \#\left\{i \colon g_i \neq 0\right\}, \quad g(x) = \sum_{i=0}^{n-1} g_i x^i,
\end{align}
and define the corresponding Hamming distance in the usual way.
\begin{lemma}\cite[Lemma 4]{bajalan2023sigma}\label{lm23}
A linear code $C \subseteq R^n$ is skew $(\sigma,\vec{a})$-polycyclic if and only if $\Omega_f^r(C)$ is a left $R[x;\sigma]$-submodule of $\mathcal{R}_f^r$.
\end{lemma}
If $f$ is central, then $\mathcal{R}_f^r$ is an associative ring. Consequently, Lemma~\ref{lm23} yields the following ideal-theoretic characterization.
\begin{corollary}
Let $f(x)$ be central. Then a linear code $C \subseteq R^n$ is skew $(\sigma, \vec{a})$-polycyclic if and only if $\Omega_f^r(C)$ is a left ideal of $\mathcal{R}_f^r$.
\end{corollary}
Throughout the remainder of this paper, we restrict attention to skew $(\sigma,f)$-polycyclic codes defined by trinomials of the form
\begin{align}
    f(x)=x^n-(a_\ell x^\ell+a_0) \in R[x;\sigma], \quad 1\leqslant \ell\leqslant n-1.
\end{align}
In order to deal with these codes as left ideals of $\mathcal{R}_{f}^{r} $, we will always assume that the trinomial $f(x)$ is central.
 We define    the set of binomials of degree $\ell$, non-zero constant term and coefficients in $(R^\sigma)^\times$ as
\begin{align}
    B_{\ell,\sigma}:=\{a(x)=a_\ell x^\ell + a_0 \colon a_0, a_\ell \in (R^\sigma)^\times\},
\end{align}
endowed with the Schur product as an internal operation
\begin{align}\label{schur}
(a_\ell x^\ell+a_0)\star (b_\ell x^\ell+b_0):=a_\ell b_\ell x^\ell+a_0 b_0.
\end{align}
Then, $(B_{\ell,\sigma},\star)$ is an abelian group with identity $e(x)=x^\ell+1$ and the inverse of an element given by
\begin{align}
    (b_\ell x^\ell+b_0)^{-1}=b_\ell^{-1}x^\ell+b_0^{-1}.
\end{align}
Since both $b_\ell,b_0 \in (R^\sigma)^\times$ and $\sigma$ is an automorphism, we have $b_\ell^{-1},b_0^{-1} \in (R^\sigma)^\times$. Therefore, $(b_\ell x^\ell+b_0)^{-1}\in B_{\ell,\sigma}$. We also define the set of binomials given by
\begin{align}
    H_{\ell,\sigma}:=\left\{N_{n-\ell}^{\sigma}(\sigma^\ell(\alpha)) x^\ell + N_n^\sigma(\alpha)\colon\alpha \in R^\times\right\}.
\end{align}

% ====================================================================
% Hamming equivalence of skew polycyclic codes
% ====================================================================

\section{Hamming equivalence of skew $(f, \sigma)$-polycyclic codes}\label{eq}
In this section, we define an equivalence relation on trinomials and show that it admits an algebraic description in terms of the set $H_{\ell,\sigma}$. Based on that equivalence, we will define the Hamming $(n,\sigma)$-equivalence of skew polycyclic codes associated with central trinomials.  We then determine which skew $(\sigma,f)$-polycyclic codes are Hamming $(n,\sigma)$-equivalent to those defined by the trinomial $x^n-(x^\ell+1)$, thereby reducing the classification problem up to equivalence  to this distinguished case.

\begin{definition}[Hamming $(n,\sigma)$-isometry]
Let $f_1(x)=x^n-a(x)$ and $f_2(x)=x^n-b(x)$ be central trinomials in $R[x;\sigma]$. A map $\varphi\colon\mathcal{R}^r_{f_1} \to \mathcal{R}^r_{f_2} $ is called a \emph{Hamming $(n, \sigma)$-isometry} if 
$\varphi$ is a ring isomorphism that preserves the Hamming distance:
\begin{align}
    d_H\left(\varphi(g), \varphi(h)\right)=d_H\left(g, h\right),\quad \text{ for all}\,\, g, h\in \mathcal{R}^r_{f_1}.
\end{align}    
\end{definition}

\begin{definition}[Hamming $(n,\sigma)$-equivalence of codes]
 Let $f_1(x)=x^n-a(x)$ and $f_2(x)=x^n-b(x)$ be central trinomials in $R[x;\sigma]$. Let $C_1 \subseteq \mathcal{R}_{f_1}^r$ be a skew $(\sigma,f_1)$-polycyclic code  and $C_2 \subseteq \mathcal{R}_{f_2}^r$ be a skew $(\sigma,f_2)$-polycyclic code. The codes $C_1$ and $C_2$ are \emph{Hamming $(n,\sigma)$-equivalent} if there exists a Hamming $(n,\sigma)$-isometry $\varphi:\mathcal{R}_{f_1}^r \longrightarrow \mathcal{R}_{f_2}^r$ such that $\varphi(C_1)=C_2$.   
\end{definition}

\begin{definition}[Hamming $(n,\sigma)$-equivalence of binomials]\label{Def_Iso}
Let $f_1(x)=x^n-a(x)$ and $f_2(x)=x^n-b(x)$ be central trinomials in $R[x;\sigma]$. The binomial $a(x)$ is $(n, \sigma)$-equivalent to the binomial $b(x)$, denoted by $a(x) \sim_{(n,\sigma)} b(x)$, if there exists a unit $\alpha \in R^\times$ such that the map
\begin{align}
    \varphi_{\alpha} \colon\mathcal{R}^r_{f_1} &\longrightarrow \mathcal{R}^r_{f_2}\nonumber \\
    \sum_{i=0}^{n-1} g_i x^i &\longmapsto \sum_{i=0}^{n-1} g_i N_i^\sigma(\alpha) x^i \label{isometry}
\end{align}
is a Hamming $(n, \sigma)$-isometry.  
\end{definition}

\begin{remark}\label{Re_1}
Since $(\alpha x)^i = N_i^\sigma(\alpha)x^i$ by Lemma \ref{norm-pro}, 
$
    \varphi_\alpha(g(x))=g(\alpha x)$ for all $  g(x)\in \mathcal{R}_{f_1}^r.
$ 
\end{remark}

\begin{remark}\label{new-re}
For $0\leqslant i\leqslant n-1$, $N_i^\sigma(\alpha)\in R^\times$ by Lemma~\ref{norm-pro}. Hence, $g_i=0$ if and only if $g_iN_i^\sigma(\alpha)=0$.
Therefore, $\varphi_\alpha$ preserves Hamming weight, and hence Hamming distance.
\end{remark}

\begin{proposition}
The relation $\sim_{(n,\sigma)}$ is an equivalence relation.
\end{proposition}

\begin{proof}
Let $a(x), b(x), c(x)$ be such that $x^n-a(x)$, $x^n-b(x)$, and $x^n-c(x)$ are central trinomials in $R[x;\sigma]$.
For proving reflexivity, take $\alpha=1$. Then $N_i^\sigma(1)=1$ for all $i$, so $\varphi_1=\id$, and hence $a(x)\sim_{(n,\sigma)} a(x)$. 
In order to check symmetry, suppose $a(x)\sim_{(n,\sigma)} b(x)$ via $\alpha\in R^\times$. Then $\varphi_\alpha$ is a Hamming $(n,\sigma)$-isometry. For $g(x)=\sum_{i=0}^{n-1} g_i x^i$, Lemma~\ref{norm-pro} yields
\begin{align}
    (\varphi_{\alpha^{-1}} \circ\varphi_\alpha)(g(x)) = \sum_{i=0}^{n-1} g_i N_i^\sigma(\alpha)N_i^\sigma(\alpha^{-1})x^i = g(x).
\end{align}
Thus $\varphi_{\alpha^{-1}}=(\varphi_\alpha)^{-1}$ is a ring isomorphism. Moreover, for $h(x)=\sum_{i=0}^{n-1} h_i x^i \in \mathcal R_{x^n-b(x)}^r$,
\begin{align}
    \varphi_{\alpha^{-1}}(h(x))= \sum_{i=0}^{n-1} h_i N_i^\sigma(\alpha^{-1})x^i.
\end{align}
Since $N_i^\sigma(\alpha^{-1})\in R^\times$ for all $i$ (Lemma \ref{norm-pro}), zero coefficients are preserved, and therefore $\operatorname{wt}_H(\varphi_{\alpha^{-1}}(h))=\operatorname{wt}_H(h)$. It follows that $b(x)\sim_{(n,\sigma)} a(x)$. Finally,  for proving transitivity, suppose that $a(x)\sim_{(n,\sigma)} b(x)$ via $\alpha\in R^\times$ and $b(x)\sim_{(n,\sigma)} c(x)$ via $\beta\in R^\times$. Then, for any $g(x)=\sum_{i=0}^{n-1} g_i x^i \in \mathcal R_{x^n-a(x)}^r$, Lemma \ref{norm-pro} gives
\begin{align}
    (\varphi_\beta \circ \varphi_\alpha)(g(x))=\sum_{i=0}^{n-1} g_i N_i^\sigma(\alpha) N_i^\sigma(\beta)x^i=\varphi_{\alpha\beta}(g(x)).
\end{align}
Thus $\varphi_\beta \circ \varphi_\alpha = \varphi_{\alpha\beta}$, which is a Hamming $(n,\sigma)$-isometry. Therefore $a(x)\sim_{(n,\sigma)} c(x)$.
\end{proof}

The next theorem provides a complete algebraic characterization of Hamming $(n,\sigma)$-equivalence of central trinomials in terms of the subgroup $H_{\ell,\sigma}$.
\begin{theorem}\label{Th.1}
Let $x^n-a(x)$ and $x^n-b(x)$ be central trinomials in $R[x;\sigma]$, where $a(x)=a_\ell x^\ell+a_0$ and $ b(x)=b_\ell x^\ell+b_0$  are in  $ B_{\ell,\sigma}$. Then, the following statements are equivalent:
\begin{enumerate}
    \item $ a(x) \sim_{(n,\sigma)} b(x)$; 
    \item   There exists  $\alpha\in R^\times$  such that $a_0 = b_0 N_n^{\sigma}(\alpha)$ and $a_{\ell}= b_{\ell} N^{\sigma}_{n-\ell}(\sigma^{\ell}(\alpha))$; 
    \item  There exists $\alpha\in R^\times$  such that $b(x)\star \left( N_{n-\ell}^{\sigma}(\sigma^{\ell} (\alpha)) x^{\ell} +  N_n^{\sigma}(\alpha)\right) = a(x)$;
    \item  $a(x)\star b(x)^{-1} \in H_{\ell,\sigma}$. 
 % H_{\ell,\sigma}:= \{ N_{n-\ell}^{\sigma}(\sigma^{\ell} (\alpha)) x^{\ell} +  N_n^{\sigma}(\alpha) ~:~ \alpha \in R^\times\}
\end{enumerate}
\end{theorem}

\begin{proof}\
\begin{itemize}
    \item[] \rm(1) $\Rightarrow$ \rm(2). Let $a(x) \sim_{(n,\sigma)} b(x)$. Then, there exists $\alpha \in R^{\times}$ such that the map $ \varphi_{\alpha} \colon\mathcal{R}^r_{x^n - a(x)}\to \mathcal{R}^r_{x^n - b(x)}$ given by $\sum_{i=0}^{n-1} g_i x^i \mapsto \sum_{i=0}^{n-1} g_i N_i^\sigma(\alpha) x^i$ is an isomorphism. Since $\varphi_\alpha$ is well-defined,  $\varphi_\alpha(x^n - a_\ell x^\ell - a_0) = 0$ in the quotient ring $\mathcal{R}^r_{x^n - b(x)}$. We note that
    \begin{align}
        \varphi_\alpha(x^n - a_\ell x^\ell - a_0) = N_n^\sigma(\alpha) x^n - a_\ell N_\ell^\sigma(\alpha) x^\ell - a_0.
    \end{align}
    Applying the relation $x^n = b_\ell x^\ell + b_0$ in the quotient ring $ \mathcal{R}^r_{x^n - b(x)}$, we obtain
    \begin{align}\label{well_def}
        N_n^\sigma(\alpha) (b_\ell x^\ell + b_0) - a_\ell N_\ell^\sigma(\alpha) x^\ell - a_0 = 0,
    \end{align}
    which implies
    \begin{align}
        \left(N_n^\sigma(\alpha) b_\ell - a_\ell N_\ell^\sigma(\alpha) \right) x^\ell + \left( N_n^\sigma(\alpha) b_0 - a_0 \right)= 0.
    \end{align}
    Hence, $a_0 = N_n^\sigma(\alpha) b_0$ and $a_\ell= b_\ell N_n^\sigma(\alpha) (N_\ell^\sigma(\alpha))^{-1}$. Applying Lemma \ref{norm-pro} gives
    \begin{align}
        a_\ell = b_\ell N_{n-\ell}^\sigma\left(\sigma^\ell(\alpha)\right)N_\ell^\sigma(\alpha)(N_\ell^\sigma(\alpha))^{-1}= b_\ell N_{n-\ell}^{\sigma}(\sigma^{\ell}(\alpha)).
    \end{align}

    \item[] \rm(2) $\Rightarrow$ \rm(3). If $h(x) := N_{n- \ell}^\sigma(\sigma^\ell(\alpha))x^\ell + N_n^\sigma(\alpha)$, then
    \begin{align*}
        b(x) \star h(x)= \left( b_\ell N_{n-\ell}^\sigma(\sigma^\ell(\alpha)) \right)x^\ell + \left( b_0 N_n^\sigma(\alpha) \right)=a_\ell x^\ell + a_0 = a(x).
    \end{align*}

    \item[] \rm(3) $\Rightarrow$ \rm(4). This is straightforward.

    \item[] \rm(4) $\Rightarrow$ \rm(1).  Suppose that $ b^{-1}(x)\star a(x)\in H_{\ell,\sigma} $.  Then, there exists $\alpha\in R^\times$ such that 
    \begin{align}
        b_{\ell}^{-1}a_{\ell}x^{\ell}+b_0^{-1}a_0= b^{-1}(x)\star a(x) =N_{n-\ell}^{\sigma}(\sigma^{\ell} (\alpha)) x^{\ell} +  N_n^{\sigma}(\alpha).
    \end{align}
    Hence
    \begin{align}\label{eq-co}
        a_{\ell}=b_{\ell} N_{n-\ell}^{\sigma}(\sigma^{\ell} (\alpha))\quad \text{and}\quad a_0= b_0 N_n^{\sigma}(\alpha).
    \end{align}
    Let $\pi: R[x;\sigma]\to \mathcal R^r_{x^n-b(x)}$ be the canonical projection and define $\Phi_\alpha: R[x;\sigma]\to R[x;\sigma]$ via $g(x)\mapsto g(\alpha x)$. By Lemma~\ref{hom_le}, $\Phi_\alpha$ is a ring homomorphism. Set $\widetilde{\varphi}_\alpha := \pi\circ \Phi_\alpha$. Then $\widetilde{\varphi}_\alpha$ is a ring homomorphism. By  \eqref{eq-co} and Lemma \ref{norm-pro},  we obtain
    \begin{align}
        \widetilde{\varphi}_{\alpha} (x^n - a_{\ell} x^{\ell} - a_0) & = N^{\sigma}_n(\alpha) x^n - N^{\sigma}_{\ell}(\alpha) a_{\ell} x^{\ell} - a_0\nonumber \\
        & = N^{\sigma}_n(\alpha) x^n - N^{\sigma}_{\ell}(\alpha)  N_{n-\ell}^{\sigma}( \sigma^{\ell}(\alpha)) b_{\ell} x^{\ell} - N^{\sigma}_n(\alpha) b_0\nonumber\\
        & = N^{\sigma}_n(\alpha) (x^n - b_{\ell} x^{\ell} - b_0)= 0 \mod \mathcal{R}_{x^n-b(x)}^r
    \end{align}
     Therefore, $\langle x^n - a_{\ell} x^{\ell} - a_0\rangle_l \subseteq \ker\widetilde{\varphi}_{\alpha}$. Conversely, let $h(x)\in \ker\widetilde{\varphi}_{\alpha}$. Then, $h(\alpha x)=0 \mod \mathcal{R}_{x^n-b(x)}^r$, so there exists $g(x)\in R[x;\sigma]$ such that $h(\alpha x)=g(x)\left(x^{n}- b_{\ell} x^{\ell} - b_0\right)$. Replacing  $x$ by $\alpha^{-1}x$ and applying Lemma \ref{norm-pro} yields
    \begin{align}
        h(x)&=g(\alpha^{-1}x)\Big( N^{\sigma}_n(\alpha^{-1}) x^{n}- N^{\sigma}_{\ell}(\alpha^{-1}) b_{\ell} x^{\ell} - b_0\Big)\nonumber \\
        & =N^{\sigma}_n(\alpha^{-1}) g(\alpha^{-1}x)\Big(x^{n}- N^{\sigma}_{n-\ell}(\sigma^{\ell}(\alpha)) b_{\ell} x^{\ell} - N^{\sigma}_n(\alpha) b_0\Big)\nonumber\\
        &= N^{\sigma}_n(\alpha^{-1}) g(\alpha^{-1}x)\left( x^n-  a_{\ell} x^{\ell} - a_0\right),
    \end{align}
    where \eqref{eq-co} is used in the last equality.
    Therefore, $\ker\widetilde{\varphi}_{\alpha} = \langle x^{n}-  a_{\ell} x^{\ell} - a_0 \rangle_l$. The first isomorphism theorem induces the ring isomorphism
    \begin{align}
        \varphi_{\alpha} \colon \mathcal{R}^r_{x^n - a(x)}  & \longrightarrow  \mathcal{R}^r_{x^n - b(x)},\nonumber \\ 
        g(x)=\sum_{i=0}^{n-1} g_i x^i &\longmapsto \sum_{i=0}^{n-1} g_i N_i^\sigma(\alpha) x^i=g(\alpha x),
    \end{align}
     which, by Remark \ref{new-re}, is a Hamming $(n,\sigma)$-isometry. Hence $a(x)\sim_{(n,\sigma)} b(x)$.
\end{itemize}
\end{proof}

\begin{example}\label{example 1}
The ring $R=\mathbb{Z}_{8}[u] / \langle u^2 - 2, 4u\rangle$ is a chain ring consisting of elements of the form $ a + b u$, where $a \in \mathbb{Z}_{8}$, $b \in \mathbb{Z}_{4}$ and $u^2 - 2=0$. The set of unit elements of $R$ is $\{a + b u \mid a\in U(\mathbb Z_8)\, ,\, b \in \mathbb Z_4\}$. The automorphism $\sigma(a+bu)=a+3bu$ of $R$ has order $2$. Take $n=4$ and $\ell=2$. The polynomials $x^4-(7x^2+1)$ and   $x^4-(x^2+1)$ are central by Lemma \ref{Cor_Asigma}. Let $\alpha=1+u\in R^\times$, and set
\begin{align}
    a(x)=a_2x^2+a_0=7x^2+1, \quad b(x)=b_2x^2+b_0=x^2+1.
\end{align}
Then $a(x),b(x)\in B_{2,\sigma}$, and $N_2^\sigma(\alpha)=7$ and $N_4^\sigma(\alpha)=1$.  Hence,
\begin{align}
    a_0=1=b_0N_4^\sigma(\alpha),\quad a_2=7=b_2N_2^\sigma(\sigma^2(\alpha)).
\end{align}
Thus condition {\rm(2)} of Theorem~\ref{Th.1} holds, so $a(x)\sim_{(4,\sigma)} b(x)$.
\end{example}

The next result identifies precisely which skew $(\sigma,f)$-polycyclic codes are Hamming $(n,\sigma)$-equivalent to those defined by the trinomial $x^n - (x^\ell + 1)$.
\begin{theorem}\label{Equiv_1}
Let $f_1(x)=x^n-a(x)$ and $f_2(x)=x^n-(x^{\ell}+1)$ be central trinomials in $R[x;\sigma]$, where $a(x)=a_\ell x^\ell+a_0\in B_{\ell,\sigma}$. Then, the following statements are equivalent:
\begin{enumerate}
    \item $a(x) \sim_{(n,\sigma)} x^\ell + 1$;
 
    \item  The class of skew $(\sigma,f_1)$-polycyclic codes is Hamming $(n,\sigma)$-equivalent to the class of skew $(\sigma,f_2)$-polycyclic codes via a Hamming $(n,\sigma)$-isometry $\varphi:\mathcal R_{f_1}^r \to \mathcal R_{f_2}^r$ such that $\varphi(r)=r$ for all $r\in R$.
 \end{enumerate}
\end{theorem}

\begin{proof}\
\begin{itemize}
    \item[] \rm(1) $\Rightarrow$ \rm(2). By Definition~\ref{Def_Iso}, there exists a Hamming $(n,\sigma)$-isometry $\varphi_\alpha:\mathcal R_{f_1}^r\longrightarrow \mathcal R_{f_2}^r$. Let $C_1\subseteq \mathcal R_{f_1}^r$ be a skew $(\sigma,f_1)$-polycyclic code. Since $f_1$ and $f_2$ are central, skew $(\sigma,f_i)$-polycyclic codes are precisely the left ideals of $\mathcal R_{f_i}^r$. Since $\varphi_\alpha$ is a ring isomorphism, it maps left ideals of $\mathcal R_{f_1}^r$ bijectively onto left ideals of $\mathcal R_{f_2}^r$. Hence, $C_2:=\varphi_\alpha(C_1)$ is a skew $(\sigma,f_2)$-polycyclic code. Now  $\varphi_\alpha$ is a Hamming $(n,\sigma)$-isometry, thus the codes $C_1$ and $C_2$ are Hamming $(n,\sigma)$-equivalent. Therefore, the statement in \rm(2) holds.

    \item[] \rm(2) $\Rightarrow$ \rm(1). By \rm(2), there exists a Hamming \((n,\sigma)\)-isometry $\varphi:\mathcal R_{f_1}^r\longrightarrow \mathcal R_{f_2}^r$. Define
    \begin{align}
         \Psi:=(\Omega_{f_2}^r)^{-1}\circ \varphi\circ \Omega_{f_1}^r:R^n\to R^n.
    \end{align}
    Since $\Omega_{f_1}^r$ and $\Omega_{f_2}^r$ are left $R$-module isomorphisms from $R^n$ onto $\mathcal{R}_{f_1}^r$ and $\mathcal{R}_{f_2}^r$ respectively, and since $\varphi$ is a ring isomorphism fixing $R$ pointwise, $\Psi$ is a left $R$-module isomorphism from the free module $R^n$ to itself. 
    Note that the skew polynomial multiplication in $\mathcal{R}_{f_i}^r$ affects the ring structure but not the left $R$-module structure, so $\mathcal{R}_{f_i}^r \cong R^n$ as left $R$-modules via $\Omega_{f_i}^r$. Moreover, $\Psi$ preserves Hamming weight since $\Omega_{f_i}^r$ is a coefficient-wise identification and $\varphi$ preserves Hamming distance by hypothesis. By the MacWilliams Extension Theorem over finite commutative chain rings \cite{greferath2014macwilliams}, $\Psi$ is therefore a monomial transformation. Thus, there exist a permutation $\pi$ of $\{0,1,\dots,n-1\}$ and units $u_i\in R^\times$ such that
    \begin{align}\label{yadesh}
        \Psi(e_i)=u_i e_{\pi(i)},\quad 0\leqslant i\leqslant n-1 ,
    \end{align}
    where $e_i$ is the vector with $1$ in the $i$-th position and zero elsewhere. Since $\Omega_{f_1}^r(e_1)=x$, we obtain
    \begin{align}\label{eqE1}
        \varphi(x)=\varphi(\Omega_{f_1}^r(e_1))=\Omega_{f_2}^r(\Psi(e_1)).
    \end{align}
    Hence, by \eqref{yadesh} and \eqref{eqE1}, there exist an element $\alpha\in R^\times$ and an index $k\in\{0,1,\dots,n-1\}$ such that
    \begin{align}\label{eqjed}
        \varphi(x)=\alpha x^k.
    \end{align}
    We claim that $k\neq 0$. Indeed, if $k=0$, then $\varphi(x)=\alpha\in R$ by \eqref{eqjed}, and hence $x=\varphi^{-1}(\alpha)\in R$, which is impossible. Therefore, by \eqref{eqjed},
    \begin{align}\label{eqphi-x}
        \varphi(x)=\alpha x^k,\quad 1\leqslant k\leqslant n-1.
    \end{align}

    For each $0\leqslant i\leqslant n-1$, by \eqref{yadesh},
    \begin{align}\label{eq-varphi-xi}
        \varphi(x^i)=\varphi(\Omega_{f_1}^r(e_i))=\Omega_{f_2}^r(\Psi(e_i))=u_i x^{\pi(i)},
    \end{align}
    where $u_i\in R^\times$ and $\pi$ is a permutation of $\{0,1,\dots,n-1\}$. Hence
    \begin{align}\label{tatko}
         \operatorname{wt}_H\bigl(\varphi(x^i)\bigr)=1,\quad \text{for all} \quad 0\leqslant i\leqslant n-1.
    \end{align}

    We will assume, for contradiction, that $k>1$. Define $ j_0:=n-k$ and $j_1:=n-k+1$. Then $j_0,j_1\in\{0,1,\dots,n-1\}$ and $j_0\neq j_1$. Since $\pi$ is a permutation, there exist unique indices $i_0,i_1\in\{0,1,\dots,n-1\}$ such that $\pi(i_0)=j_0$ and $\pi(i_1)=j_1$. As $i_0\neq i_1$, at most one of them equals $n-1$. Hence one of them lies in $\{0,1,\dots,n-2\}$. Choose $i\in\{0,1,\dots,n-2\}$ such that $\pi(i)\in\{j_0,j_1\}$ and write $j:=\pi(i)$. Since $i \leqslant n-2$, we have $i+1 \leq n-1$, so $x^{i+1} = x \cdot x^i$ holds without reduction in $\mathcal{R}_{f_1}^r$, and therefore $\varphi(x^{i+1}) = \varphi(x)\varphi(x^i)$. Thus, using \eqref{eqphi-x} and \eqref{eq-varphi-xi}, we obtain
    \begin{align}
     \varphi(x^{i+1})=\varphi(x)\varphi(x^i)=(\alpha x^k)(u_i x^j)=\alpha \sigma^k(u_i) x^{k+j}.
    \end{align}
    Since $\alpha \sigma^k(u_i)\in R^\times$, it follows that
    \begin{align}\label{mama}
        \operatorname{wt}_H\bigl(\varphi(x^{i+1})\bigr)=\operatorname{wt}_H(x^{k+j})
    \end{align}
    in $\mathcal R_{f_2}^r$, where $x^n=x^\ell+1$.

    If $j=j_0$, then $x^{k+j}=x^n=x^\ell+1$, so $\operatorname{wt}_H(x^{k+j})=2$ (since $\ell\geqslant 1$).
    If $j=j_1$, then $x^{k+j}=x^{n+1}=x^{\ell+1}+x$. If $\ell+1<n$, then $\operatorname{wt}_H(x^{k+j})=2$. If $\ell=n-1$, then $ x^{k+j}=x^{\ell+1}+x=x^{\ell}+1+x$, so $\operatorname{wt}_H(x^{k+j})=3$. Thus, in all cases,
    \begin{align}
        \operatorname{wt}_H(x^{k+j})\geqslant 2.
    \end{align}
    Hence, by \eqref{mama},
    \begin{align}\label{sa}
         \operatorname{wt}_H\bigl(\varphi(x^{i+1})\bigr)\geqslant 2.
    \end{align}
    On the other hand, since $i\leqslant n-2$, we have $i+1\le n-1$, and therefore by \eqref{tatko},
    \begin{align}
        \operatorname{wt}_H\bigl(\varphi(x^{i+1})\bigr)=1,
    \end{align}
    which contradicts \eqref{sa}. Hence $k=1$.

    Therefore, by \eqref{eqphi-x}, $\varphi(x)=\alpha x$ for some $\alpha\in R^\times$. Consequently, by Lemma \ref{norm-pro},
    \begin{align}
        \varphi(x^i)=(\varphi(x))^i=(\alpha x)^i=N_i^\sigma(\alpha)x^i,\quad 0\leqslant i\leqslant n-1.
    \end{align}
    Hence, since $\varphi(r) = r$ for all $r \in R$ by hypothesis, we have
    \begin{align}
        \varphi\Bigl(\sum_{i=0}^{n-1} g_i x^i\Bigr)=\sum_{i=0}^{n-1} g_i N_i^\sigma(\alpha)x^i.
    \end{align}
    If we define $\varphi_\alpha:=\varphi$, then Definition \ref{Def_Iso} leads to $a(x)\sim_{(n,\sigma)}x^\ell+1$.  
\end{itemize}
\end{proof}

 If one applies Theorem~\ref{Equiv_1} and then Theorem~\ref{Th.1} with $b(x)=x^\ell+1$, by considering that $x^\ell+1$ is the identity element of the group $(B_{\ell,\sigma},\star)$, it follows 
the corollary below which shows that the study of skew $(\sigma,f)$-polycyclic codes can be reduced to the case $x^n - (x^\ell + 1)$, with all equivalent families obtained via the 
subgroup $H_{\ell,\sigma}$.

\begin{corollary}\label{vital-eq2}
 Let $f_1(x)=x^n-a(x)$ and $f_2(x)=x^n-(x^\ell+1)$ be central trinomials with  $a(x)\in B_{\ell,\sigma}$.  Then, the class of skew $(\sigma,\,f_1)$-polycyclic codes is Hamming $(n,\sigma)$-equivalent to the class of skew $(\sigma,\,f_2)$-polycyclic codes if and only if
$a(x)\in H_{\ell,\sigma}$.
\end{corollary}

%\begin{proof}
% We apply Theorem~\ref{Equiv_1} and then Theorem~\ref{Th.1} with $b(x)=x^\ell+1$, by considering that $x^\ell+1$ is the identity element of the group $(B_{\ell,\sigma},\star)$.   
%\end{proof}

\begin{corollary}\label{vital-eq}
Let $f_1(x)=x^n-(a_\ell x^\ell+a_0)$ and $f_2(x)=x^n-(x^\ell+1)$ be two trinomials with the following conditions:
\begin{enumerate}
    \item $n=\mu m$ for some positive integer $m$, and $\mu$ is the order of $\sigma$;
    \item $a_0,a_\ell \in (R^\sigma)^\times$;
    \item $\ell \equiv 0 \pmod{\mu}$.
\end{enumerate}
Then, the class of skew $(\sigma,\,f_1)$-polycyclic codes is Hamming $(n,\sigma)$-equivalent to the class of skew $(\sigma,\,f_2)$-polycyclic codes if and only if
$a_\ell x^\ell+a_0\in H_{\ell,\sigma}$.
\end{corollary}

\begin{proof}
By Lemma~\ref{Cor_Asigma}, the conditions {\rm{1}}-{\rm3{}} are sufficient to ensure that both trinomials $f_1(x)$ and $f_2(x)$ are central. Furthermore,  $a(x):=a_\ell x^\ell+a_0$ belongs to $B_{\ell,\sigma}$ by the condition {\rm{2}}. Now, we apply Corollary \ref{vital-eq2}.
\end{proof}
% ====================================================================
% Number of Hamming equivalence classes
% ====================================================================
\section{Number of Hamming equivalence classes}\label{sec:number}

In this section, we count the number of skew $(\sigma,f_1)$-polycyclic codes that are Hamming $(n,\sigma)$-equivalent to skew $(\sigma,f_2)$-polycyclic codes, where $f_1(x)=x^n-a(x)$ and $f_2(x)=x^n-(x^\ell+1)$ are central trinomials with $a(x)\in B_{\ell,\sigma}$. By Corollary~\ref{vital-eq2}, this class consists precisely of all $a(x)\in H_{\ell,\sigma}$, and hence has cardinality $|H_{\ell,\sigma}|$. We therefore begin by determining $|H_{\ell,\sigma}|$.

\begin{lemma}
Let $\sigma \in \operatorname{Aut}(R)$ be an automorphism of order $\mu$, and suppose that $n\equiv \ell \equiv 0 \pmod{\mu}$. Then $H_{\ell,\sigma}$ is a subgroup of $(B_{\ell,\sigma},\star)$.
\end{lemma}

\begin{proof}
Define
\begin{align}\label{ker-theta}
\theta : R^\times &\longrightarrow B_{\ell,\sigma}, \nonumber\\
\alpha &\longmapsto N_{n-\ell}^{\sigma}(\sigma^\ell(\alpha)) x^\ell + N_n^\sigma(\alpha).
\end{align}
Since $\sigma^\ell=\sigma^n=\id_R$, Lemma~\ref{norm-pro} implies that 
$N_n^\sigma(\alpha)\in (R^\sigma)^{\times}$ and $N_{n-\ell}^\sigma(\sigma^\ell(\alpha)) \in (R^\sigma)^{\times}$ for all $\alpha\in R^\times$. Hence $\theta$ is well-defined. Moreover, for $\alpha,\beta\in R^\times$, Lemma~\ref{norm-pro} yields
\begin{align}
\theta(\alpha\beta)= N_{n-\ell}^\sigma(\sigma^\ell(\alpha))\,N_{n-\ell}^\sigma(\sigma^\ell(\beta)) x^\ell
+ N_n^\sigma(\alpha)\,N_n^\sigma(\beta) =\theta(\alpha)\star \theta(\beta).
\end{align}
Thus $\theta$ is a group homomorphism. Therefore, $H_{\ell,\sigma}=\operatorname{Im}(\theta)$ is a subgroup of $(B_{\ell,\sigma},\star)$.
\end{proof}

In the following example, we show that if at least one of the conditions in the lemma fails, the set $H_{\ell,\sigma}$ is not necessarily a subgroup of $(B_{\ell,\sigma}, \star)$.

\begin{example}
Consider the Galois ring $R=\mathrm{GR}(4,2)\cong \mathbb{Z}_4[\omega]$, where $\omega$ satisfies $\omega^2+\omega+1=0$. Let $\sigma\in\operatorname{Aut}(R)$ be the Frobenius automorphism defined by $\sigma(a+b\omega)=a+b\omega^2$. Then $\sigma$ has order $\mu=2$, and its fixed subring is $R^\sigma=\mathbb{Z}_4$, so $(R^\sigma)^\times=\{1,3\}$. Consider $n=3$ and $\ell=2$. Then $n\not\equiv 0 \pmod{\mu}$. Let $\alpha=\omega\in R^\times$. We have
\begin{align}
N_{1}^\sigma(\sigma^2(\omega))=\omega \notin (R^\sigma)^\times,\quad N_3^\sigma(\omega)= \prod_{t=0}^{2}\sigma^t(\omega)= \omega \cdot \omega^2 \cdot \omega= \omega^4=\omega\notin (R^\sigma)^\times.
\end{align}
It follows that $\theta(\omega)\notin B_{\ell,\sigma}$, and hence $H_{\ell,\sigma}\not\subseteq B_{\ell,\sigma}$.
\end{example}

By \eqref{UnitGroup}, the unit group $R^\times$ decomposes as $R^\times = T^* \cdot U \cong T^* \times U$, where $U = 1 + \gamma R$ is a finite abelian $p$-group of order $p^{r(e-1)}$ and $T^*$ is a subgroup of order $p^r - 1$. The following lemma gives an explicit formula for the $\sigma$-norm on $T^*$. 

\begin{lemma}\label{lem:AutTeich}
Let $\sigma \in \aut(R)$ be an automorphism. Then the following holds:
\begin{enumerate}
    \item There exists an integer $k \in \{0, 1, \dots, r-1\}$ such that $\sigma(\xi) = \xi^{p^k}$ for all $\xi \in T^*$.
    \item For $i \geqslant 0$ and $\xi \in T^*$, the $\sigma$-norm of $\xi$ is given by
    \begin{align}
         N_i^\sigma(\xi)=\begin{cases}\xi^{\frac{p^{ik}-1}{p^k-1}}, &\text{if } p^k \neq 1, \\\xi^i, & \text{if } p^k =1.\end{cases}
    \end{align}
    
    \item For $\xi \in T^*$, the equality $N_i^\sigma(\xi)=1$ holds if and only if $\operatorname{ord}(\xi) \mid \gcd\!\left(\frac{p^{ik}-1}{p^k-1},\, q-1\right)$ when $p^k \ne 1$, or $\operatorname{ord}(\xi) \mid \gcd(i,\, q-1)$ when $p^k = 1$.
\end{enumerate}
\end{lemma}

\begin{proof}\
\begin{enumerate}
    \item  The group $\operatorname{Aut}(\mathbb{F}_{p^r})$ is cyclic, generated by the Frobenius automorphism. Hence the induced automorphism $\overline{\sigma}$ is of the form $\overline{\sigma}(\bar{x})=\bar{x}^{p^k}$ for some $k\in\{0,\dots,r-1\}$. The canonical projection $T^*\to \mathbb{F}_q^\times$, $\xi\mapsto \overline{\xi}$, is a bijection because it is a surjection between two finite sets of the same size. For any $\xi\in T^*$, we have
    \begin{align}\label{mama}
        \overline{\sigma(\xi)}=\overline{\sigma}(\overline{\xi})=(\overline{\xi})^{p^k}=\overline{\xi^{p^k}}.
    \end{align}
    By Lemma~\ref{stabilize}, if $\xi \in T^*$, then $\sigma(\xi) \in T^*$.  Thus, both $\sigma(\xi)$ and $\xi^{p^k}$ belong to $T^*$. Since the canonical projection is an isomorphism, and hence injective, \eqref{mama} implies that $\sigma(\xi)=\xi^{p^k}$.

    \item By the definition of the $\sigma$-norm and part~(1), we obtain
    \begin{align}
    N_i^\sigma(\xi)= \sigma^{i-1}(\xi)\,\sigma^{i-2}(\xi) \cdots \sigma(\xi)\, \xi= \xi^{(p^k)^{i-1}}\, \xi^{(p^k)^{i-2}} \cdots \xi^{p^k}\, \xi^1=\xi^{\sum_{j=0}^{i-1} (p^k)^j}.
    \end{align}
   The exponent is a geometric sum with $i$ terms and common ratio $p^k$, which yields the desired formula.

   \item Since $T^*$ is a cyclic group of order $q-1$, the condition $\xi^m=1$ is equivalent to $\operatorname{ord}(\xi)\mid \gcd(m,q-1)$. The result follows immediately.
\end{enumerate}
\end{proof}

Since $U=1+\gamma R$ is a finite abelian $p$-group, the structure theorem for finite abelian groups implies that there exist uniquely determined integers $k_1,\dots,k_J\geqslant 1$, up to permutation, and elements $u_1,\dots,u_J\in U$, not necessarily unique, such that
\begin{align}\label{im61}
    U=\langle u_1\rangle\times\cdots\times\langle u_J\rangle \cong \mathbb Z_{p^{k_1}}\times\cdots\times \mathbb Z_{p^{k_J}},
\end{align}
where each $u_i$ has order $p^{k_i}$. Thus, every element $u\in U$ can be written uniquely as
\begin{align}\label{im62}
    u=u_1^{m_1}\cdots u_J^{m_J},\quad m_i\in \mathbb Z_{p^{k_i}}.
\end{align}
Recall the group homomorphism $\theta$ defined in \eqref{ker-theta}. Since the identity of $(B_{\ell,\sigma},\star)$ is $x^\ell+1$, its kernel is
\begin{align}
    \ker(\theta)=\left\{\alpha\in R^\times \colon N_{n-\ell}^{\sigma}(\sigma^\ell(\alpha))=1 \text{ and } N_n^\sigma(\alpha)=1\right\}.
\end{align}
The next theorem determines the cardinality of $\ker(\theta)$ and, consequently, that of $H_{\ell,\sigma}$. 

\begin{theorem}\label{proper-all}
Let $\sigma \in \operatorname{Aut}(R)$ be an automorphism of order $\mu$, and suppose that $n\equiv \ell \equiv 0 \pmod{\mu}$. Let $k$ be a non-negative integer as in Lemma~\ref{lem:AutTeich}. For $1\leqslant i\leqslant J$, set $a_i:=N_{n-\ell}^\sigma(u_i)$ and $b_i:=N_n^\sigma(u_i)$. Then:
\begin{enumerate}
    \item The cardinality of $\ker(\theta)$ is 
    \begin{align}
        |\ker(\theta)| = |\ker(\theta|_{T^*})| \cdot |\ker(\theta|_{U})|,
    \end{align}
    where 
    \begin{align}\label{result1}
    |\ker(\theta|_{T^*})|=
    \begin{cases}
    \gcd\!\left(\dfrac{p^{(n-\ell)k}-1}{p^k-1},\,\dfrac{p^{nk}-1}{p^k-1},\,q-1\right), & \text{if } p^k\neq 1,\\[1.2ex]
    \gcd(n,\,n-\ell,\,q-1), & \text{if } p^k=1.
    \end{cases}
    \end{align}
    and 
    \begin{align}
    |\ker(\theta|_U)|=\#\left\{(m_1,\dots,m_J)\in \bigoplus_{i=1}^J \mathbb Z_{p^{k_i}} :
    \prod_{i=1}^J a_i^{m_i}=1,\ 
    \prod_{i=1}^J b_i^{m_i}=1\right\};
    \end{align}

    \item The cardinality of $H_{\ell,\sigma}$ is
    \begin{align}\label{general-H}
    |H_{\ell,\sigma}|=\frac{p^{r(e-1)}(p^{r}-1)}{|\ker(\theta)|}.
    \end{align}
\end{enumerate}
\end{theorem}

\begin{proof}\
 \begin{enumerate}
     \item Since $\sigma^\ell=\operatorname{Id}_R$, $\theta(\alpha)=N_{n-\ell}^\sigma(\alpha)x^\ell+N_n^\sigma(\alpha)$ for all $\alpha\in R^\times$. Moreover, since $\mu\mid (n-\ell)$, we have $\sigma^n=\sigma^{n-\ell}=\operatorname{Id}_R$, and hence by Lemma \ref{norm-pro},
     \begin{align}\label{stabi}
         N_n^\sigma(\alpha)\,\,, N_{n-\ell}^\sigma(\alpha)\in R^\sigma\quad \text{for all}\quad \alpha\in R.
     \end{align}
     Since $R^\times = T^* \cdot U$, every element $\alpha\in R^\times$ can be written uniquely as $\alpha=\xi u$ for some $\xi\in T^*$ and $u\in U$. By Lemma \ref{stabilize},  $N_n^\sigma(\xi),\,N_{n-\ell}^\sigma(\xi)\in T^*$ and $N_n^\sigma(u),\,N_{n-\ell}^\sigma(u)\in U$. Hence, \eqref{stabi} implies that
    \begin{align}\label{nie}
        N_n^\sigma(\xi),\, N_{n-\ell}^\sigma(\xi)\in T^*\cap R^\sigma,\quad N_n^\sigma(u),\, N_{n-\ell}^\sigma(u)\in U\cap R^\sigma.
    \end{align}
    Thus 
    \begin{align}\label{stabi2}
      N_n^\sigma(\xi) N_n^\sigma(u)\in (T^*\cap R^\sigma)\cdot (U\cap R^\sigma),\quad
      N_{n-\ell}^\sigma(\xi) N_{n-\ell}^\sigma(u) \in  (T^*\cap R^\sigma)\cdot (U\cap R^\sigma). 
    \end{align}
    Since $(R^\sigma)^\times = (T^*\cap R^\sigma)\cdot (U\cap R^\sigma)$ and $(T^*\cap R^\sigma)\cap (U\cap R^\sigma)=\{1\}$ every element of $(R^\sigma)^\times$ admits a unique factorization of this form. In particular, if $N_n^\sigma(\xi) N_n^\sigma(u)=1$ then by \eqref{nie}, we have
    \begin{align}
        N_n^\sigma(\xi)=\bigl(N_n^\sigma(u)\bigr)^{-1} \in (T^*\cap R^\sigma)\cap (U\cap R^\sigma)=\{1\},
    \end{align}
    and hence $N_n^\sigma(\xi)=1$ and $N_n^\sigma(u)=1$. The same argument applies to $N_{n-\ell}^\sigma$, and therefore
    \begin{align}
        N_n^\sigma(\xi)\,N_n^\sigma(u)=1
        &\Longleftrightarrow N_n^\sigma(\xi)=1 \text{ and } N_n^\sigma(u)=1,\label{imp23}\\
        N_{n-\ell}^\sigma(\xi)\,N_{n-\ell}^\sigma(u)=1
        &\Longleftrightarrow N_{n-\ell}^\sigma(\xi)=1 \text{ and } N_{n-\ell}^\sigma(u)=1.\label{imp22}
    \end{align} 
    On the other hand, since $\theta$ is a group homomorphism, we have $\theta(\alpha)=\theta(\xi)\star\theta(u)$. Thus, the definition of $\theta$ and relations \eqref{imp23} and \eqref{imp22} yield
    \begin{align}
        \theta(\alpha)=x^\ell+1\quad\Longleftrightarrow\quad  N_n^\sigma(\xi)= N_n^\sigma(u)=1= N_{n-\ell}^\sigma(\xi)= N_{n-\ell}^\sigma(u).
    \end{align}
        Therefore, $\alpha\in\ker(\theta)$ if and only if $\xi\in\ker(\theta|_{T^*})$ and $u\in\ker(\theta|_U)$. Since the decomposition $\alpha=\xi u$ with $\xi\in T^*$ and $u\in U$ is unique, the map $\ker(\theta)\rightarrow \ker(\theta|_{T^*})\times \ker(\theta|_U)$ given by $\xi u\longmapsto (\xi,u)$ is a well-defined bijection, and hence $ \ker(\theta)\cong \ker(\theta|_{T^*}) \times \ker(\theta|_{U})$. Therefore,
    \begin{align}
       |\ker(\theta)| = |\ker(\theta|_{T^*})| \cdot |\ker(\theta|_{U})|.
    \end{align}
    We first compute $|\ker(\theta|_{T^*})|$. Let $\xi\in T^*$. Then, $\xi\in \ker(\theta|_{T^*})$ if and only if $N_{n-\ell}^\sigma(\xi)=1$ and $N_n^\sigma(\xi)=1$. Assume first that $p^k\neq 1$, and set $A=(p^{(n-\ell)k}-1)/(p^k-1),$ and $B=(p^{nk}-1)/(p^k-1)$. By Lemma~\ref{lem:AutTeich}, $\ker(\theta|_{T^*})=\{\xi\in T^*:\xi^A=1 \text{ and } \xi^B=1\}$. Applying B\'ezout's identity, and applying the fact that $\gcd(A,B)$ divides both $A$ and $B$, we have an element $\xi\in T^*$ satisfies $\xi^A=1$ and $\xi^B=1$ if and only if $\xi^{\gcd(A,B)}=1$.  Therefore, $\ker(\theta|_{T^*})=\{\xi\in T^*:\xi^{\gcd(A,B)}=1\}$. In a cyclic group of order $q-1$, the number of solutions of $x^m=1$ is $\gcd(m,q-1)$. It follows that
    \begin{align}
        |\ker(\theta|_{T^*})|=\gcd(\gcd(A,B),q-1)=\gcd(A,B,q-1).
    \end{align}
    If $p^k=1$, the same argument yields $|\ker(\theta|_{T^*})|=\gcd(n,n-\ell,q-1)$. Therefore, we have \eqref{result1}.
    
    We next compute $|\ker(\theta|_{U})|$. From \eqref{im62}, every element $u\in U$ can be written uniquely in the form
    \begin{align}
        u=u_1^{m_1}\cdots u_J^{m_J},\quad m_i\in \mathbb Z_{p^{k_i}}.
    \end{align}
    Lemma \ref{norm-pro} yields
    \begin{align}
        N_{n-\ell}^\sigma(u)=\prod_{i=1}^J (N_{n-\ell}^\sigma(u_i))^{m_i}=\prod_{i=1}^J a_i^{m_i},\quad N_n^\sigma(u)=\prod_{i=1}^J (N_n^\sigma(u_i))^{m_i}=\prod_{i=1}^J b_i^{m_i}.
    \end{align}
    Therefore, $u\in\ker(\theta|_U)$ if and only if $\prod_{i=1}^J a_i^{m_i}=1$ and  $\prod_{i=1}^J b_i^{m_i}=1$. Since the representation of $u$ by the tuple $(m_1,\dots,m_J)$ is unique, the required cardinality is exactly the number of tuples satisfying these two relations, that is   
    \begin{align}
        |\ker(\theta|_U)|=\#\left\{(m_1,\dots,m_J)\in \bigoplus_{i=1}^J \mathbb Z_{p^{k_i}}\colon \prod_{i=1}^J a_i^{m_i}=1,\ \prod_{i=1}^J b_i^{m_i}=1\right\}.
    \end{align}

    \item The second statement follows from the First Isomorphism Theorem, and also from the cardinality $|R^\times|=p^{r(e-1)}(p^r-1)$
 \end{enumerate}
\end{proof}

We are now in a position to apply the results of this section to enumerate all central trinomials whose associated skew polycyclic codes are Hamming $(n,\sigma)$-equivalent to those defined by the central trinomial $x^n-(x^\ell+1)$.

\begin{theorem}
Let $\sigma \in \operatorname{Aut}(R)$ be an automorphism of order $\mu$, and suppose that $n\equiv \ell \equiv 0 \pmod{\mu}$. Let $x^n-(x^\ell+1)$
be a central trinomial in $R[x;\sigma]$. Then the number of central trinomials
\[
f(x)=x^n-a(x), \quad\text{with}\quad a(x)\in B_{\ell,\sigma},
\]
such that the class of skew $(\sigma,f)$-polycyclic codes is Hamming $(n,\sigma)$-equivalent to the class of skew $(\sigma,\, x^n-(x^\ell+1))$-polycyclic codes is
\begin{align}
|H_{\ell,\sigma}|=\frac{p^{r(e-1)}(p^r-1)}{|\ker(\theta)|},
\end{align}
where $|\ker(\theta)|$ is given by Theorem~\ref{proper-all}.
\end{theorem}

\begin{proof}
The result follows directly from Corollary~\ref{vital-eq2} and Theorem~\ref{proper-all}.
\end{proof}

The following example illustrates the computation of $|H_{\ell,\sigma}|$ and the enumeration of Hamming $(n,\sigma)$-equivalent skew polycyclic codes.

\begin{example}
Let $R$ be the chain ring in Example \ref{example 1} with automorphism $\sigma(a+bu)=a+3bu$, which has order $2$. Take $n=4$ and $\ell=2$. We compute $|H_{2,\sigma}|$.

First, we determine the structure of the group $U=1+\gamma R$. Each element of $U$ takes the form $\alpha=1+2k+bu$, where $k, b\in \mathbb Z_4$, and therefore $|U|=16$.  A direct computation using $u^2=2$ and $4u=0$ shows that $\alpha^4=1$ for every $\alpha\in U$.%\footnote{By the property of cyclic groups and group theory, if $g^n = e$ for some integer $n$, then the order of the element $g$ must divide $n$. Recall that the exponent of a group $U$ is defined as the smallest positive integer $m$ such that $g^m = e$ for every $g \in U$.}. Therefore, the only possible values for the exponent are $1, 2,$ or $4$.  
Since $(1+u)^2=3+2u\neq 1$ and $(1+u)^4=1$, the exponent of $U$ is exactly $4$. Hence, by the classification theorem for finite abelian groups %\footnote{Every finite abelian group $U$ of order $16$ ($2^4$) can be expressed as a product of cyclic groups: $U \cong \mathbb{Z}_{2^{k_1}} \times \mathbb{Z}_{2^{k_2}} \times \cdots \times \mathbb{Z}_{2^{k_n}}$ where the sum of the exponents $k_1 + k_2 + \dots + k_n = 4$.The exponent of a direct product is $\lcm\{k_1, k_2, \ldots, k_n\}$. Since the exponent of $U$ is exactly $4$, the largest cyclic component must be $\mathbb{Z}_4$.}
\begin{align}
   U\cong \mathbb Z_4\times \mathbb Z_4 \quad\text{or}\quad U\cong\mathbb Z_4\times \mathbb Z_2\times \mathbb Z_2.
\end{align}
To distinguish these possibilities, we count the elements of order dividing
$2$. For $\alpha=a+bu\in U$, with $a\in\{1,3,5,7\}$ and $b\in\mathbb Z_4$,
we have $\alpha^2=a^2+2b^2+2abu$. Thus $\alpha^2=1$ if and only if $a^2+2b^2\equiv 1 \pmod 8$ and $2ab\equiv 0 \pmod 4$. Since $a$ is odd, this is equivalent to $b\in\{0,2\}$. Hence, there are exactly
$8$ elements of $U$ satisfying $\alpha^2=1$. Since
$\mathbb Z_4\times\mathbb Z_4$ has only $4$ such elements, we conclude that $ U\cong \mathbb Z_4\times \mathbb Z_2\times \mathbb Z_2$. To find an explicit decomposition of $U$, we set $x=1+u$, $y=1+2u$, and $z=5$. Then $x^4=y^2=z^2=1$, and
\begin{align}
    \langle x\rangle=\{1,\ 1+u,\ 3+2u,\ 3+3u\},\quad \langle y,z\rangle=\{1,\ 5,\ 1+2u,\ 5+2u\}.
\end{align}
Moreover, $\langle x\rangle\cap\langle y,z\rangle=\{1\}$, so $|\langle x\rangle\langle y,z\rangle|=|\langle x\rangle|\,|\langle y,z\rangle|=16$. Since $|U|=16$, we obtain
%Since $\langle x\rangle\langle y,z\rangle\subseteq U$ and $|U|=16$, we conclude that $U=\langle x\rangle\langle y,z\rangle$. Finally, $\langle y,z\rangle$ has order $4$, and since $y$ and $z$ are distinct elements of order $2$, we have $\langle y,z\rangle=\langle y\rangle\times \langle z\rangle$
\begin{align}
    U=\langle x\rangle\times \langle y\rangle\times \langle z\rangle\cong \mathbb Z_4\times \mathbb Z_2\times \mathbb Z_2.
\end{align}

Next, we compute $|\ker(\theta|_U)|$. We have $\sigma(x)=1-u$, $\sigma(y)=1+2u$, $\sigma(z)=5$. Since $N_2^\sigma(u)=\sigma(u)u$ for any $u\in U$, we get $N_2^\sigma(x)=7$, $N_2^\sigma(y)=1=N_2^\sigma(z)$. Since $\sigma^2=\operatorname{Id}_R$, for every $u\in U$ we have $N_4^\sigma(u)=(N_2^\sigma(u))^2$.
Hence $N_4^\sigma(x)=N_4^\sigma(y)=N_4^\sigma(z)=1$. Every element of $U$ can be written uniquely as $x^{m_1}y^{m_2}z^{m_3}$, where $(m_1,m_2,m_3)\in \mathbb Z_4\oplus\mathbb Z_2\oplus\mathbb Z_2$. Therefore,
\begin{align}
|\ker(\theta|_U)|
&=\#\left\{(m_1,m_2,m_3)\in \mathbb Z_4\oplus \mathbb Z_2\oplus \mathbb Z_2 :
7^{m_1}=1\right\}.
\end{align}
Since $7^2=1$, this gives $|\ker(\theta|_U)|=2\cdot 2\cdot 2=8$.

Next, we compute $|\ker(\theta|_{T^*})|$. Since the residue field of $R$ is $\mathbb F_2$, we have $q=2$. Hence, the Teichm\"uller subgroup has order $|T^*|=q-1=1$, and therefore $T^*=\{1\}$. It follows immediately that $|\ker(\theta|_{T^*})|=1$.

Finally, combining the above, we obtain $|\ker(\theta)|=|\ker(\theta|_{T^*})|\cdot|\ker(\theta|_U)|=8$. Since $|R^\times|=p^{r(e-1)}(p^{r}-1)=16$, Theorem \ref{proper-all} gives
\begin{align}
|H_{2,\sigma}|=\frac{|R^\times|}{|\ker(\theta)|}=\frac{16}{8}=2.
\end{align}
Thus, the number of central trinomials $f_1(x)=x^4-a(x)$ with $a(x)=a_\ell x^\ell+a_0$ such that the class of skew $(\sigma,f_a)$-polycyclic codes is Hamming $(n,\sigma)$-equivalent to the class of skew $(\sigma,f_0)$-polycyclic codes is $|H_{2,\sigma}|=2$.

\end{example}

%====================================================================
% Conclusion
% ===================================================================
\section*{Conclusion}

We studied Hamming equivalence for skew polycyclic codes over finite chain rings defined by central trinomials. We introduced an $(n,\sigma)$-equivalence relation on the defining trinomials and characterized it algebraically via the subgroup $H_{\ell,\sigma}$ of the Schur-product group $(B_{\ell,\sigma},\star)$. This yields a criterion for determining when two such code families are Hamming $(n,\sigma)$-equivalent. In particular, we showed that the family of skew $(\sigma,f)$-polycyclic codes with $f(x)=x^n-(a_\ell x^\ell+a_0)$ is Hamming $(n,\sigma)$-equivalent to the family defined by $x^n-(x^\ell+1)$ whenever $a_\ell x^\ell+a_0\in H_{\ell,\sigma}$. We also computed $|H_{\ell,\sigma}|$ using the decomposition $R^\times=T^*\cdot(1+\gamma R)$, thereby enumerating the corresponding central trinomials. Future work may consider analogous equivalence notions for the rank and sum-rank metrics.

\printbibliography
\end{document}